\documentclass[conference,letterpaper]{IEEEtran}

\usepackage[utf8]{inputenc}
\usepackage[T1]{fontenc}
\usepackage{url}
\usepackage{ifthen}
\usepackage{cite}
\usepackage[cmex10]{amsmath} 

\usepackage{ieee}
\usepackage[]{flushend} 

\begin{document}
\title{
Strong Converse Exponent for Remote Lossy Source Coding
}

\author{%
  \IEEEauthorblockN{Han Wu and Hamdi Joudeh}
  \IEEEauthorblockA{Eindhoven University of Technology, The Netherlands \\
                    Email: \{h.wu1, h.joudeh\}@tue.nl }
}

\maketitle


\begin{abstract}
  Past works on remote lossy source coding studied the rate under average distortion and the error exponent of excess distortion probability.
  In this work, we look into how fast the excess distortion probability converges to 1 at small rates, also known as exponential strong converse.
  We characterize its exponent by establishing matched upper and lower bounds.
  From the exponent, we also recover two previous results on lossy source coding and biometric authentication.
\end{abstract}

\section{Introduction}
Let \(X^n\) be a discrete memoryless source that represents some real-world data.
Given a reconstruction alphabet \(\hat{\mathcal{X}}\) and an additive distortion measure \(d(x, \hat{x})\), Shannon \cite{shannonCodingTheoremsDiscrete1959} established the celebrated rate-distortion function \(R(\Delta)\) of lossy source coding: the minimum compression rate such that  \(X^n\) can be asymptotically reconstructed within an average distortion level \(\Delta\).
Following Shannon's work, Dobrushin and Tsybakov \cite{dobrushinInformationTransmissionAdditional1962} introduced remote lossy source coding, where the task remains to reconstruct \(X^n\) but the compressor now only has access to its noisy observation \(Y^n\).
Remote lossy source coding, also known as indirect (or noisy) lossy source coding, is of interest because it incorporates the random nature of measurement taken by electrical equipments, such as fingerprint scanners.
Dobrushin and Tsybakov \cite{dobrushinInformationTransmissionAdditional1962} showed that the rate-distortion function in this setting is the same as Shannon's rate-distortion function for \(Y^n\) under the distortion \(\tilde{d}(y, \hat{x}) = \E[ d(X, \hat{x}) | y ]\).
More works along this direction can be found in, e.g., \cite{wolfTransmissionNoisyInformation1970, bergerRateDistortionTheory1971, witsenhausenIndirectRateDistortion1980}.

Besides the average distortion, another important distortion criterion for lossy source coding is the excess distortion.
Let \(\hat{X}^n\) be the reconstructed source.
The excess distortion concerns the probability of source sequences that are reconstructed beyond distortion level \(\Delta\), i.e., \( \P \{ \bar{d}( X^n, \hat{X}^n) > \Delta \} \), where \(\bar{d}\) denotes the per-letter distortion.
Marton \cite{martonErrorExponentSource1974} established the error exponent of excess distortion probability for lossy source coding.
Weissman and Merhav \cite{weissmanTradeoffsExcesscodelengthExponent2002} characterized the error exponent of excess distortion probability for remote lossy source coding, where they also studied a variable-length codeword variant of the problem.
Later on, Weissman extended these results to the universal source coding case in \cite{weissmanUniversallyAttainableError2004}.

In this work, we are interested in \( \P \{ \bar{d}( X^n, \hat{X}^n) \leq \Delta \} \) for remote lossy source coding.
In particular, we look into how fast it decays to \(0\) at small rates, which serves as the exponential strong converse for remote lossy source coding.
We will characterize the exponent, establish its continuity, and also show that the exponent is strictly positive for all rates below the rate-distortion function.
Two previous results are also recovered, including Csiszár and Körner \cite[Problem 9.6]{csiszarInformationTheoryCoding2011} for lossy source coding and Kang \etal \cite[Theorem 1]{kangDeceptionSideInformation2015} for biometric authentication.
Moreover, our converse proof serves as a simpler alternative proof for \cite[Theorem 1]{kangDeceptionSideInformation2015}.
Proofs in this work are based on the method of types \cite{csiszarInformationTheoryCoding2011}.

Remote lossy source coding has found wide applications since its introduction, and some variants have been considered, e.g., multiterminal \cite{yamamotoSourceCodingTheory1980}, the CEO problem \cite{bergerCEOProblem1996,oohamaRatedistortionFunctionQuadratic1998}, Gaussian \cite{oohamaIndirectDirectGaussian2014,eswaranRemoteSourceCoding2019}.
Nonasymptotic analysis for remote lossy source coding is provided in \cite{kostinaNonasymptoticNoisyLossy2016}, and \(f\)-separable distortion measures are studied in \cite{stavrouIndirectRateDistortion2023}.
Resurgent interests in remote lossy source coding are partly driven by its connection to the information bottleneck problem under logarithmic loss \cite{zaidiInformationBottleneckProblems2020,goldfeldInformationBottleneckProblem2020}.
Another area that received renewed attention is exponential strong converse \cite{arimotoConverseCodingTheorem1973}.
Oohama established the exponential strong converse for several multiterminal source coding problems through the information spectrum method in, e.g., \cite{oohamaExponentialStrongConverse2018,oohamaExponentialStrongConverse2019}.

We briefly mention some notation used throughout the work.
For a finite alphabet \(\mathcal{X}\), let \(\mathcal{P}(\mathcal{X})\) be the set of all pmfs defined on \(\mathcal{X}\).
We write \(x^n\) or \(\bm{x}\) for an \(n\)-length sequence from \(\mathcal{X}^n\).
Random vectors are denoted by \(X^n\) or \(\bm{X}\), which also applies to \(Y^n\) or \(\bm{Y}\).
All alphabets are assumed to be finite.
The type of a sequence \(\bm{x}\) is denoted by \(\hat{P}_{\bm{x}}\), while \(\hat{P}_{\bm{x} \bm{y}}\) is the joint type and \(\hat{P}_{\bm{y} | \bm{x}}\) is the conditional type.
The set of all types is written as \(\mathcal{P}_n(\mathcal{X})\).
A type class \(\mathcal{T}_n(P_X)\) consists of all \(\bm{x}\) satisfying \(\hat{P}_{\bm{x}} = P_X\), while \(\mathcal{T}_n(Q_{Y|X} | \bm{x})\) denotes the conditional type class, i.e., all \(\bm{y}\) satisfying \(\hat{P}_{\bm{y} | \bm{x}} = Q_{Y|X}\).
Let \(a_n \ndot{\geq} b_n\) if \(\liminf_{n \to \infty} \frac{1}{n}\log (a_n/b_n) \geq 0\) and \(a_n \ndot{\leq} b_n\) if \(\limsup_{n\to \infty}\frac{1}{n}\log (a_n/b_n) \leq 0\).
We will write \(a_n \ndot{=} b_n\) if both hold.
Let \(|a|^{+} \triangleq \max \{0, a\}\) and \([N] \triangleq \{1,2,\cdots, N \}\).
The base of exponential and log functions is chosen to be \(2\).

\section{Problem Setup and Main Results}
Consider a pmf \(P_{XY} \in \mathcal{P}( \mathcal{X} \times \mathcal{Y} )\).
Assume that we have a DMS pair \((X^n, Y^n)\) following the distribution
\begin{equation}
    P_{X^n Y^n} ( x^n, y^n) = \prod_{i=1}^{n}P_{XY}(x_i, y_i).
\end{equation}
We can interpret \(X^n\) as a remote source and \(Y^n\) as its noisy observation through the discrete memoryless channel \(P_{Y|X}\).
A helper has access to the noisy observation \(Y^n\) and describes it to a receiver through a rate-limited link \(f_n:\mathcal{Y}^n \to [2^{nR}]\).
Given a forwarded index \(m \in [2^{nR}]\), the receiver tries to reconstruct the source sequence through a decoder \(g_n: [2^{nR}] \to \hat{\mathcal{X}}^n \), where \(\hat{\mathcal{X}}\) is the reconstruction alphabet.
The distortion between a source sequence \(x^n\) and its reconstruction \(\hat{x}_n\) is measured through
\begin{equation}
  \bar{d}(x^n, \hat{x}^n) \triangleq \frac{1}{n} \sum_{i=1}^{n} d(x_i, \hat{x}_i),
\end{equation}
where \(d: \mathcal{X} \times \hat{\mathcal{X}} \to \mathbb{R}^{+} \) is a distortion mapping.

Denote the reconstructed source by \(\hat{X}^n\), i.e., \(\hat{X}_n = g_n( f_n(Y^n) )\).
For a distortion level \(\Delta\), we define
\begin{equation}
    p_{\mathrm{c}} (n, R, \Delta) \triangleq \max_{f_n, g_n} \P \{ \bar{d}(X^n, \hat{X}^n ) \leq \Delta \},
\end{equation}
i.e., the maximal correct reconstruction probability for a given rate-distortion pair \((R, \Delta)\).
Note that throughout the work we assume \(\Delta \geq \Delta_{\min}\) where \(\Delta_{\min} \triangleq \min_{\hat{x}(y)} \E [d(X, \hat{x}(Y)  ]\).
In this work, we are interested in how fast \(p_{\mathrm{c}} (n, R, \Delta)\) decays to 0 as \(n\) grows, captured by
\begin{equation}
    E(R, \Delta) \triangleq \lim_{n \to \infty} -\frac{1}{n} \log p_{\mathrm{c}} (n, R, \Delta).
\end{equation}
We will give a full characterization of \(E(R, \Delta)\) for all pairs \((R, \Delta)\) and also show that \(E(R, \Delta) > 0\) if and only if \(R\) is below the rate-distortion function \(R_{\mathrm{r}}(P_{XY},\Delta)\), where
\begin{equation}
  R_{\mathrm{r}}(P_{XY} , \Delta) \triangleq \! \! \! \! \min_{ \substack{ P_{\hat{X} | Y}: \\ \E[ d(X, \hat{X}) ] \leq \Delta}  } \! \! \! \! I(Y;\hat{X}),
\end{equation}
in which we have the joint distribution \(P_{XY\hat{X}} = P_{XY} P_{ \hat{X} | Y }\).

\subsection{Main Results}
Before stating the main result, we need to define
\begin{equation*}
    E(Q_{XY}, R, \Delta) \triangleq \! \! \! \! \min_{ \substack{ Q_{\hat{X} | X Y}: \\ \E_Q[ d(X, \hat{X}) ] \leq \Delta}  } \! \! \! \! I_Q(X;\hat{X} | Y) + |I_Q(Y;\hat{X}) - R|^{+},
\end{equation*}
where we have the joint distribution \(Q_{XY\hat{X}} = Q_{XY} Q_{ \hat{X} | XY }\).

\begin{theorem}
\label{thm:reliability_function_remote_lossy_source_coding}
Given a DMS pair \(P_{XY}^n\), we have
\begin{equation}
    E(R, \Delta) = \min_{Q_{XY}} D(Q_{XY} \| P_{XY}) + E(Q_{XY}, R, \Delta ). \label{eq:characterization_reliability_function_remote_lossy_source_coding}
\end{equation}
\end{theorem}
\begin{proof}
  See Section \ref{sec:achievability} and Section \ref{sec:converse}.
\end{proof}

We next look into the continuity of \(E(R, \Delta)\) and also the  rate regime for which \(E(R, \Delta)>0\).
\begin{proposition}
\label{prop:continuity_and_rate_regime}
For a fixed \(\Delta\), \(E(R, \Delta)\) is a continuous function of \(R\) and \(E(R, \Delta) > 0\) for every rate \(R < R_{\mathrm{r}}(P_{XY}, \Delta)\).
\end{proposition}
\begin{proof}
  See Section \ref{sec:rate_regime}.
\end{proof}
In \cite{kangDeceptionSideInformation2015}, Kang \emph{et al.} investigated a deception problem in the context of biometric authentication, which is equivalent to \(R \geq \log |\hat{\mathcal{X}}|\) in the setting of remote lossy source coding.
Hence we can recover their result \cite[Theorem 1]{kangDeceptionSideInformation2015} from Theorem \ref{thm:reliability_function_remote_lossy_source_coding}.
For this purpose, define the conditional rate-distortion function as
\begin{equation}
    R_{\mathrm{c}}(Q_{XY}, \Delta) \triangleq \! \! \! \! \min_{ \substack{ Q_{\hat{X} | X Y}: \\ \E_Q[ d(X, \hat{X}) ] \leq \Delta}  } \! \! \! \! I_Q(X;\hat{X} | Y),
\end{equation}
where we have the joint distribution \(Q_{XY\hat{X}} = Q_{XY} Q_{ \hat{X} | XY }\).
\begin{corollary}
  \label{cor:no_compression}
For all \(R \geq \log |\hat{\mathcal{X}}|\), we have
\begin{equation}
    E(R, \Delta) = \min_{Q_{XY}} D(Q_{XY} \| P_{XY}) + R_{\mathrm{c}}(Q_{XY}, \Delta).
\end{equation}
\end{corollary}
Corollary \ref{cor:no_compression} is proved by noting that \(E(Q_{XY}, R, \Delta ) = R_{\mathrm{c}}(Q_{XY}, \Delta)\) when \(R \geq \log |\hat{\mathcal{X}}|\).
The converse proof provided in \cite[Theorem 1]{kangDeceptionSideInformation2015} involves a sophisticated codebook construction as well as the blowing-up lemma.
Our converse proof in Section \ref{sec:converse} exploits the method of types and serves as a simpler alternative converse proof for \cite[Theorem 1]{kangDeceptionSideInformation2015}.

If \(Y^n\) is a noiseless observation of \(X^n\), i.e., \(P_{XY}\) satisfies \(\P\{ X = Y \} =1\), then remote lossy source coding is reduced to the standard lossy source coding problem.
Hence, we can also recover \cite[Problem 9.6]{csiszarInformationTheoryCoding2011}  from Theorem \ref{thm:reliability_function_remote_lossy_source_coding}.
To this end, define the standard rate-distortion function as
\begin{equation}
    R(Q_{X}, \Delta) \triangleq \! \! \! \! \min_{ \substack{ Q_{\hat{X} | X}: \\ \E_Q[ d(X, \hat{X}) ] \leq \Delta}  } \! \! \! \! I_Q(X;\hat{X}).
\end{equation}
\begin{corollary}
    If \(P_{XY}\) satisfies \(\P\{ X = Y \} =1\), then we have
    \begin{equation}
      E(R, \Delta) = \min_{Q_{X}} D(Q_{X} \| P_{X}) + | R(Q_X, \Delta) - R |^{+}.
    \end{equation}
\end{corollary}
\begin{proof}
If \(\P\{ X = Y \} =1\), then the minimization over \(Q_{\hat{X} | X Y}\) in  \(E(Q_{XY}, R, \Delta)\) is  reduced to the minimization over \(Q_{\hat{X} | X}\).
Moreover, we also have \( I(X; \hat{X} | Y) = 0\) in this case.
The proof is completed after  passing the minimization over \(Q_{\hat{X} | X}\) into the term \( | I_Q( Y ; \hat{X}) - R |^{+}\) in \(E(Q_{XY}, R, \Delta)\).
\end{proof}

\section{Achievability}
\label{sec:achievability}
In this section, we show that for remote lossy source coding there exists a sequence of coding schemes \((f_n, g_n)\) such that
\begin{equation}
  \P \{ \bar{d}(X^n, \hat{X}^n ) \leq \Delta \} \ndot{\geq} 2^{-nE(R, \Delta)},
\end{equation}
where \(E(R, \Delta)\) is as stated in \eqref{eq:characterization_reliability_function_remote_lossy_source_coding}.
Before starting, we first state the well-know type covering lemma, which is originally due to Berger \cite{bergerRateDistortionTheory1971}. The version we present here appears in other literature, e.g., \cite[Lemma 3.34]{moserAdvancedTopicsInformation2022a}.
\begin{lemma}
    \label{lem:type_covering_lemma}
    For every joint type \(Q_{Y \hat{X}}\), there exists a subset \(\mathcal{A}_n \subseteq \mathcal{T}_n( Q_{\hat{X}} )\) with
    \begin{equation}
      \abs{\mathcal{A}_n} \ndot{\leq} 2^{nI(Q_{Y}, Q_{\hat{X} | Y})}
    \end{equation}
    such that for every \(\bm{y} \in \mathcal{T}_n(Q_Y)\) we can find a \(\hat{\bm{x}} \in \mathcal{A}_n\) satisfying \(\hat{P}_{ \bm{y} \hat{\bm{x}}} = Q_{Y \hat{X}}\).
\end{lemma}
The encoder and decoder operate on a type-per-type basis with respect to the observed sequence \(\bm{y}\), accomplished by sending an additional index indicating the type \(\hat{P}_{\bm{y}}\) of \(\bm{y}\).
Since there are at most \((n+1)^{|\mathcal{Y}|}\) types, adding this type index does not break the rate-limit \(R\) asymptotically.
Moreover, we allow the encoder to send one more index \(i=0\) besides \(i \in [2^{nR}]\) (whose purpose will become clear further on).
In a similar manner, this does not break the rate-limit \(R\) either.
Bearing these and the type covering lemma in mind, now we are in a position to present the coding scheme.

\subsection{Coding Scheme}
For every \(Q_Y \in \mathcal{P}_n(\mathcal{Y})\), we select a conditional type \(Q_{\hat{X} | Y}\).
Note that \(Q_{\hat{X} | Y}\) can vary for different \(Q_Y\).
From Lemma \ref{lem:type_covering_lemma}, we find a codebook \(\mathcal{A}_n = ( \hat{\bm{x}}_1, \hat{\bm{x}}_2, \cdots )\) consisting of \(M_n \ndot{=} 2^{nI(Q_{Y}, Q_{\hat{X} | Y})}\) codewords such that \(\mathcal{A}_n\) covers the type class \(\mathcal{T}_n(Q_Y)\) under \(Q_{\hat{X} | Y}\).
According to \(\mathcal{A}_n\), we partition \(\mathcal{T}_n(Q_Y)\) into \(M_n\) subregions \(\mathcal{F}_1, \mathcal{F}_2, \cdots, \mathcal{F}_{M_n}\) such that for every \(i \in [M_n]\), the codeword \(\hat{\bm{x}}_i\) and set \(\mathcal{F}_i\) satisfy
\begin{equation*}
    \hat{P}_{\bm{y} \hat{\bm{x}}_i } = Q_{Y \hat{X}}, \qquad \forall \bm{y} \in \mathcal{F}_i.
\end{equation*}
Further, assume without loss of generality that the sizes of \(\mathcal{F}_i\) are decreasing, i.e.,
\begin{equation}
    |\mathcal{F}_1| \geq |\mathcal{F}_2|  \geq \cdots \geq |\mathcal{F}_{M_n}|. \label{eq:achievability/size_decreasing}
\end{equation}
Depending on whether \(Q_{\hat{X} | Y}\) satisfies \(I(Q_Y, Q_{\hat{X} | Y}) \leq R\) or not, the coding scheme proceeds as follows.
\begin{enumerate}
    \item Assume \(I(Q_Y, Q_{\hat{X} | Y}) \leq R\).
    For every \(\bm{y} \in \mathcal{T}_n(Q_Y)\), if \(\bm{y} \in \mathcal{F}_i\), then the encoder sends \(i\), and the decoder reconstructs \(\hat{\bm{x}}_i\) according to \(\mathcal{A}_n\).
    \item Assume \(I(Q_Y, Q_{\hat{X} | Y}) > R\).
    For every \(\bm{y} \in \mathcal{T}_n(Q_Y)\), if \(\bm{y} \in \mathcal{F}_i\) where \(i \leq 2^{nR}\), then the encoder sends \(i\) and the decoder reconstructs \(\hat{\bm{x}}_i\) according to \(\mathcal{A}_n\).
    However, if \(\bm{y} \in \mathcal{F}_i\) where \(i > 2^{nR}\), then the encoder sends \(i=0\) and the decoder reconstructs an arbitrary sequence.
\end{enumerate}
The above procedures are repeated for every type \(Q_Y\).

\subsection{Exponent Analysis}
For a type class \(\mathcal{T}_n(Q_Y)\), if the selected \(Q_{\hat{X} | Y}\) results in \(I(Q_{Y}, Q_{\hat{X} | Y}) > R\), then there exist some \(\bm{y} \in \mathcal{T}_n(Q_Y)\) such that \(\bm{y} \in \mathcal{F}_i\) where \(i > 2^{nR}\), and hence \(f_n(\bm{y}) = 0\).
Since we are looking at \(\P \{ \bar{d}(X^n, \hat{X}^n ) \leq \Delta \}\), these \(\bm{y}\) leading to \(f_n(\bm{y}) = 0\) are out of our interest, and we simply discard such \(\bm{y}\) and declare an error.
Hence, we need to investigate the set
\begin{equation}
    \mathcal{H}_n(Q_Y) \triangleq \{ \bm{y} \in \mathcal{T}_n (Q_Y) | f_n(\bm{y}) \in [2^{nR}]\},
\end{equation}
particularly, the size of \(\mathcal{H}_n(Q_Y)\).
If the selected \(Q_{\hat{X} | Y}\) satisfies \(I(Q_{Y}, Q_{\hat{X} | Y}) \leq R\), then \(f_n(\bm{y}) = 0\) does not occur and
\begin{equation}
    |\mathcal{H}_n(Q_Y)| = |\mathcal{T}_n(Q_Y)|. \label{eq:achievability/lower_bound_set_size_less_than_R}
\end{equation}
On the other hand, if \(I(Q_{Y}, Q_{\hat{X} | Y}) > R\), then  we have
\begin{equation}
    \mathcal{H}_n(Q_Y) = \bigcup_{i=1}^{2^{nR}} \mathcal{F}_i,
\end{equation}
Due to the partitioning, it holds that
\begin{equation}
    | \mathcal{H}_n(Q_Y) | = \sum_{i=1}^{2^{nR}} |\mathcal{F}_i|. \label{eq:achievability/sum_size}
\end{equation}
Our aim here is to lower bound \(\P \{ \bar{d}(X^n, \hat{X}^n ) \leq \Delta \}\), and for this it is desirable to find a lower bound for \eqref{eq:achievability/sum_size}.
Observe that
\begin{equation}
    \frac{1}{2^{nR}} \sum_{i=1}^{2^{nR}} |\mathcal{F}_i| \geq \frac{1}{M_n}\sum_{i=1}^{M_n} |\mathcal{F}_i|, \label{eq:achievability/lower_bound_size_through_average}
\end{equation}
since in the RHS of \eqref{eq:achievability/lower_bound_size_through_average} we include smaller sets in the average (recall \eqref{eq:achievability/size_decreasing}).
Thus, from \eqref{eq:achievability/lower_bound_size_through_average}, we obtain
\begin{align}
    | \mathcal{H}_n(Q_Y) | & \geq 2^{n(R - \frac{1}{n}\log M_n)} \sum_{i=1}^{M_n} |\mathcal{F}_i| \\
    & \ndot{ = }  2^{n(R - I_Q(Y;\hat{X}) )} |\mathcal{T}_n(Q_Y)|, \label{eq:achievability/lower_bound_set_size_greater_than_R}
\end{align}
where \eqref{eq:achievability/lower_bound_set_size_greater_than_R} is due to \(M_n \ndot{=} 2^{nI_Q(Y;\hat{X})}\) and the partitioning.
Combining \eqref{eq:achievability/lower_bound_set_size_less_than_R} and \eqref{eq:achievability/lower_bound_set_size_greater_than_R}, no matter whether \(I_Q(Y;\hat{X}) > R\) or not, we see that
\begin{align}
    | \mathcal{H}_n(Q_Y) | & \ndot{\geq} 2^{-n | I_Q(Y;\hat{X}) -R |^{+}} |\mathcal{T}_n(Q_Y)| \\
    & \ndot{=} 2^{-n( |I_Q(Y;\hat{X}) -R |^{+} - H(Q_Y) )}. \label{eq:achievability/unified_size_lower_bound}
\end{align}

With the above prelude, we now proceed to the main part of the proof. Under the described coding scheme, we have
\begin{align}
    & \P \{ \bar{d}(\bm{X}, g_n ( f_n ( \bm{Y})) ) \leq \Delta \} \nonumber \\
    & = \! \! \! \! \sum_{Q_Y \in \mathcal{P}_n(\mathcal{Y})} \sum_{\bm{y} \in \mathcal{T}_n(Q_Y)} \! \! \!  \P \{ \bm{Y} = \bm{y}\} \times \P \{ \bar{d}(\bm{X}, g_n ( f_n ( \bm{y})) ) \leq \Delta   \} \nonumber \\
    & \geq \! \! \! \! \sum_{Q_Y \in \mathcal{P}_n(\mathcal{Y})} \sum_{\bm{y} \in \mathcal{H}_n(Q_Y)} \! \! \!  \P \{ \bm{Y} = \bm{y}\} \times \P \{ \bar{d}(\bm{X}, g_n ( f_n ( \bm{y})) ) \leq \Delta   \}. \nonumber
\end{align}
For a given \(\bm{y} \in \mathcal{H}_n(Q_Y)\), let \(\hat{\bm{x}} = g_n(f_n(\bm{y}))\) and hence
\(\hat{P}_{\bm{y} \hat{\bm{x}}} = Q_{Y\hat{X}}\).
It follows that
\begin{align}
    \{\bm{x}  | \bar{d}(\bm{x}, g_n ( f_n ( \bm{y})) ) \leq \Delta  \}  = \! \! \! \! \! \! \bigcup_{ \substack{ Q_{X | Y \hat{X}}: \\ \E_Q[d( X, \hat{X} )] \leq \Delta} }  \! \! \! \! \! \! \! \mathcal{T}_n( Q_{X| Y \hat{X}} | \bm{y}, \hat{\bm{x}} ), \nonumber
\end{align}
where we have \(Q_{XY\hat{X}} = \hat{P}_{\bm{y} \hat{\bm{x}}} Q_{X| Y \hat{X}} = Q_{Y\hat{X}}Q_{X| Y \hat{X}}\).
Therefore, for any \(\bm{y} \in \mathcal{H}_n(Q_Y)\), we obtain
\begin{align}
    & \P \{ \bar{d}(\bm{X}, g_n ( f_n ( \bm{y})) ) \leq \Delta   \} \nonumber \\
    & \geq \max_{ \substack{ Q_{X | Y \hat{X}}: \\ \E_Q[d( X, \hat{X} )] \leq \Delta}  } \P \{ \mathcal{T}_n( Q_{X| Y \hat{X}} | \bm{y}, \hat{\bm{x}} )  \} \\
    & \ndot{\geq} \max_{ \substack{ Q_{X | Y \hat{X}}: \\ \E_Q[d( X, \hat{X} )] \leq \Delta}  } 2^{-nD( Q_{X | Y \hat{X}} \| P_{X|Y} | Q_{Y\hat{X}} ) }.
\end{align}
On the other hand, for every \(\bm{y} \in \mathcal{H}_n(Q_Y) \subset \mathcal{T}_n(Q_Y)\),
\begin{equation}
    \P \{ \bm{Y} = \bm{y}\} = 2^{-n(D(Q_Y \| P_Y) + H( Q_Y ))}.
\end{equation}
Consequently, we deduce that
\begin{align}
    & \sum_{\bm{y} \in \mathcal{H}_n(Q_Y)} \! \! \!   \P \{ \bm{Y} = \bm{y}\} \times \P \{ \bar{d}(\bm{X}, g_n ( f_n ( \bm{y})) ) \leq \Delta   \} \nonumber \\
    & \ndot{\geq} |\mathcal{H}_n(Q_Y)| \times 2^{-n(D(Q_Y \| P_Y) + H( Q_Y ))} \times \nonumber \\
    & \hspace{1.5cm} \max_{ \substack{ Q_{X | Y \hat{X}}: \\ \E_Q[d( X, \hat{X} )] \leq \Delta}  } 2^{-nD( Q_{X | Y \hat{X}} \| P_{X|Y} | Q_{Y\hat{X}} ) } \\
    & \ndot{\geq} \hspace{-10pt} \max_{ \substack{ Q_{X | Y \hat{X}}: \\ \E_Q[d( X, \hat{X} )] \leq \Delta}  } \hspace{-10pt} \exp \{ - n( D(Q_Y \| P_Y) + | I_Q(Y; \hat{X}) -R |^{+} + \nonumber \\
    & \hspace{4cm} \! D( Q_{X | Y \hat{X}} \| P_{X|Y} | Q_{Y\hat{X}} )   ) \} \label{eq:achievability/applying_unified_size_lower_bound} \\
    & = \hspace{-7pt} \max_{ \substack{ Q_{\hat{X} | Y}, Q_{X | Y \hat{X}}: \\ \E_Q[d( X, \hat{X} )] \leq \Delta}  } \hspace{-10pt} \exp \{ - n( D(Q_Y \| P_Y)   +  | I_Q(Y; \hat{X}) -R |^{+} + \nonumber \\
    & \hspace{4cm} D( Q_{X | Y \hat{X}} \| P_{X|Y} | Q_{Y\hat{X}} )  ) \} \label{eq:achievability/optimizing_over_Q_hat_X_Y}\\
    & = \hspace{-7pt} \max_{ \substack{ Q_{X | Y}, Q_{\hat{X} | Y X}: \\ \E_Q[d( X, \hat{X} )] \leq \Delta}  } \hspace{-10pt} \exp \{ - n( D(Q_Y \| P_Y)   + | I_Q(Y; \hat{X}) -R |^{+} + \nonumber \\
    & \hspace{4cm}  D( Q_{X | Y \hat{X}} \| P_{X|Y} | Q_{Y\hat{X}} )   ) \} \label{eq:achievability/equivalent_optimization} \\
    & = \hspace{-7pt} \max_{ \substack{ Q_{X | Y}, Q_{\hat{X} | Y X}: \\ \E_Q[d( X, \hat{X} )] \leq \Delta}  } \hspace{-10pt} \exp \{ - n( D(Q_{XY} \| P_{XY}) \!  + | I_Q(Y; \hat{X}) -R |^{+}  \nonumber \\
    & \hspace{5.2cm}  + I_Q(X;\hat{X} | Y)  ) \}, \label{eq:achievability/mutual_information_identity}
\end{align}
where \eqref{eq:achievability/applying_unified_size_lower_bound} is due to \eqref{eq:achievability/unified_size_lower_bound};
in \eqref{eq:achievability/optimizing_over_Q_hat_X_Y} we assume the best \(Q_{\hat{X} | Y}\) is chosen when constructing the codebooks;
in \eqref{eq:achievability/equivalent_optimization} we notice that optimizing over \( ( Q_{\hat{X} | Y}, Q_{X | Y \hat{X}} )\) is the same as optimizing over \((Q_{X | Y}, Q_{\hat{X} | Y X})\);
in \eqref{eq:achievability/mutual_information_identity} we utilize
\begin{align}
    & D(Q_Y \| P_Y)   +  D( Q_{X | Y \hat{X}} \| P_{X|Y} | Q_{Y\hat{X}} ) \nonumber  \\
    & =D(Q_Y \| P_Y)   +  D( Q_{X | Y \hat{X}} \| Q_{X|Y} | Q_{Y\hat{X}} ) + \nonumber \\
    & \hspace{4cm} D(Q_{Y|X} \| P_{Y|X} | Q_Y) \\
    & = D(Q_{XY} \| P_{XY})  +  I_Q(X;\hat{X} | Y). \label{eq:achievability/mutual_information_identity_presentation}
\end{align}
The remaining steps are to take the maximization over \(Q_Y\) into account and combine the maximization over \((Q_Y, Q_{X|Y})\) into \(Q_{XY}\), which then completes the achievability proof.

\section{Converse}
\label{sec:converse}
In this section, we show that for remote lossy source coding any sequence of coding schemes \((f_n, g_n)\) must satisfy
\begin{equation}
  \P \{ \bar{d}(X^n, \hat{X}^n ) \leq \Delta \}  \ndot{\leq} 2^{-nE(R, \Delta)},
\end{equation}
where \(E(R, \Delta)\) is as stated in \eqref{eq:characterization_reliability_function_remote_lossy_source_coding}.
Before starting, it is instructive to first highlight the main challenge of the converse proof and how we will deal with it.
For any \((f_n, g_n)\), consider a \(\bm{y} \in \mathcal{Y}^n\) and let \(\hat{\bm{x}} = g_n(f_n(\bm{y}))\).
Thus, for this \(\bm{y}\), we have
\begin{align}
    & \P \{ \bar{d}(\bm{X}, g_n ( f_n ( \bm{y})) ) \leq \Delta   \} \nonumber \\
    & \leq \sum_{ \substack{ Q_{X | Y \hat{X}}: \\ \E_Q[d( X, \hat{X} )] \leq \Delta} } \P \{ \mathcal{T}_n( Q_{X| Y \hat{X}} | \bm{y}, \hat{\bm{x}} )  \} \\
    & \ndot{\leq}  \max_{ \substack{ Q_{X | Y \hat{X}}: \\ \E_Q[d( X, \hat{X} )] \leq \Delta}  } 2^{-nD( Q_{X | Y \hat{X}} \| P_{X|Y} | \hat{P}_{\bm{y} \hat{\bm{x}}} ) }. \label{eq:converse/upper_bound_depend_on_type}
\end{align}
As we can see, the upper bound in \eqref{eq:converse/upper_bound_depend_on_type} depends on \(\hat{P}_{\bm{y} \hat{\bm{x}}}\), which may be any joint type since there is no restriction on \((f_n, g_n)\), as long as \(||f_n|| \leq 2^{nR}\).
The main challenge of the converse proof is in connecting \(\hat{P}_{\bm{y} \hat{\bm{x}}}\) to the fact that \(||f_n|| \leq 2^{nR}\).
In the achievability proof, such connection is accomplished through the type covering lemma by fixing \(\hat{P}_{\bm{y} \hat{\bm{x}}} = Q_{Y \hat{X}}\).
For the converse proof, we shall shift our perspective from observed sequence \(\bm{y}\) to codeword \(\hat{\bm{x}}\), and look into how many \(\bm{y}\) satisfy \(\hat{P}_{\bm{y} \hat{\bm{x}}} = Q_{Y \hat{X}}\) for a fixed codeword \(\hat{\bm{x}}\) and joint type \(Q_{Y \hat{X}}\).

Consider any coding scheme \((f_n, g_n)\) with \(||f_n|| \leq 2^{nR}\).
Let \(\mathcal{C}_n = ( \hat{\bm{x}}_1, \hat{\bm{x}}_2, \cdots, \hat{\bm{x}}_{2^{nR}})\) denote the reconstruction codebook of \(g_n\).
For every \(Q_Y \in \mathcal{P}_n(Q_Y)\), we define \(2^{nR}\) subregions
\begin{equation}
    \mathcal{G}_i(Q_Y) \triangleq \{ \bm{y} \in \mathcal{T}_n(Q_Y) | g_n( f_n(\bm{y}) ) = \hat{\bm{x}}_i \} \quad i \in [2^{nR}], \nonumber
\end{equation}
i.e., sequences from the type class \(\mathcal{T}_n(Q_Y)\) that map into \(\hat{\bm{x}}_i\).
It is clear that the \(2^{nR}\) sets \( \mathcal{G}_i(Q_Y)\) form a partition of \(\mathcal{T}_n(Q_Y)\).
For every subregion \(\mathcal{G}_i(Q_Y)\), we further partition it into second-level subsets by defining
\begin{equation}
    \mathcal{G}_i( Q_{Y\hat{X}} | Q_Y) \triangleq \{ \bm{y} \in \mathcal{G}_i(Q_Y) | \hat{P}_{\bm{y} \hat{\bm{x}}_i} = Q_{Y \hat{X}} \},
\end{equation}
i.e., sequences not only from the type class \(\mathcal{T}_n(Q_Y)\) and mapped into \(\hat{\bm{x}}_i\), but also satisfying \(\hat{P}_{\bm{y} \hat{\bm{x}}_i} = Q_{Y \hat{X}} \).
Every \(\mathcal{G}_i(Q_Y)\) can be divided into at most \((n+1)^{\abs{\mathcal{X}} \abs{\mathcal{Y}}}\) such subsets.

Given \(g_n( f_n(\bm{y}) ) = \hat{\bm{x}}_i\), the quantity in \eqref{eq:converse/upper_bound_depend_on_type} only depends on \((\bm{y}, \hat{\bm{x}}_i)\) through their joint type \(\hat{P}_{\bm{y} \hat{\bm{x}}_i}\).
Hence,  for a joint type \(Q_{Y \hat{X}}\), we are interested in the set
\begin{equation}
    \mathcal{G}^{\star}( Q_{Y\hat{X}} | Q_Y ) \triangleq \bigcup_{i=1}^{2^{nR}} \mathcal{G}_i( Q_{Y\hat{X}} | Q_Y), \label{eq:converse/union_subsets}
\end{equation}
which consists of all \(\bm{y} \in \mathcal{T}_n(Q_Y)\) for which \(\hat{P}_{\bm{y} \hat{\bm{x}}_i} = Q_{Y\hat{X}}\), i.e., all \(\bm{y} \in \mathcal{T}_n(Q_Y)\) that results in the same value in \eqref{eq:converse/upper_bound_depend_on_type}.
The task here is to find an upper bound for \(\P \{ \bar{d}(X^n, \hat{X}^n ) \leq \Delta \}\), so now it is useful to explore an upper bound on the size of \eqref{eq:converse/union_subsets}.
First, due to \( \mathcal{G}^{\star}( Q_{Y\hat{X}} | Q_Y ) \subseteq \mathcal{T}_n(Q_Y)\), we see that
\begin{equation}
    | \mathcal{G}^{\star}( Q_{Y\hat{X}} | Q_Y ) | \leq | \mathcal{T}_n(Q_Y)| \ndot{=} 2^{nH(Q_Y)}. \label{eq:converse/upper_bound_type_class}
\end{equation}
On the other hand, since \(\hat{P}_{\bm{y} \hat{\bm{x}}_i} = Q_{Y \hat{X}}\), we observe that
\begin{equation}
    |\mathcal{G}_i( Q_{Y\hat{X}} | Q_Y)| \leq |\mathcal{T}_n(Q_{Y | \hat{X} } | \hat{\bm{x}}_i)| \ndot{=} 2^{nH(Q_{Y | \hat{X}} | Q_{\hat{X}})}.
\end{equation}
Thus, because of the partitioning, we obtain another bound
\begin{align}
    | \mathcal{G}^{\star}( Q_{Y\hat{X}} | Q_Y ) | & = \sum_{i=1}^{2^{nR}} |\mathcal{G}_i( Q_{Y\hat{X}} | Q_Y)|  \\
    & \ndot{\leq} 2^{n( R + H(Q_{Y | \hat{X}} | Q_{\hat{X}}) )} \\
    & = 2^{n(R - I_Q(Y;\hat{X}))} \times 2^{nH(Q_Y)}. \label{eq:converse/upper_bound_union_conditional_type_class}
\end{align}
Combining \eqref{eq:converse/upper_bound_type_class} and \eqref{eq:converse/upper_bound_union_conditional_type_class}, we conclude that
\begin{equation}
    | \mathcal{G}^{\star}( Q_{Y\hat{X}} | Q_Y ) | \ndot{\leq} 2^{- n(|I_Q(Y;\hat{X}) - R |^{+} - H(Q_Y))}, \label{eq:converse/unified_size_upper_bound}
\end{equation}
which always yields the tightest upper bound between the two.

From now on, the converse proof follows in the same fashion as the achievability proof but in the opposite direction.
For any coding scheme \((f_n, g_n)\), we have
\begin{align}
    & \P \{ \bar{d}(\bm{X}, g_n ( f_n ( \bm{Y})) ) \leq \Delta \} \nonumber \\
    & = \! \! \! \! \sum_{Q_Y \in \mathcal{P}_n(\mathcal{Y})} \sum_{\bm{y} \in \mathcal{T}_n(Q_Y)} \! \! \!  \P \{ \bm{Y} = \bm{y}\} \times \P \{ \bar{d}(\bm{X}, g_n ( f_n ( \bm{y})) ) \leq \Delta   \} \nonumber \\
    & \leq \! \! \! \! \sum_{Q_Y \in \mathcal{P}_n(\mathcal{Y})} \sum_{Q_{ \hat{X} | Y}}  \sum_{\bm{y} \in \mathcal{G}^{\star}( Q_{Y\hat{X}} | Q_Y ) } \! \! \!  \P \{ \bm{Y} = \bm{y}\} \times \nonumber \\
    & \hspace{3.5cm} \P \{ \bar{d}(\bm{X}, g_n ( f_n ( \bm{y})) ) \leq \Delta   \}, \label{eq:converse/sum_over_conditional_type}
\end{align}
where in \eqref{eq:converse/sum_over_conditional_type} we have \(Q_{Y\hat{X}} = Q_Y \times Q_{\hat{X} | Y}\).
For any \(\bm{y} \in \mathcal{G}^{\star}( Q_{Y\hat{X}} | Q_Y )\), following from \eqref{eq:converse/upper_bound_depend_on_type}, we see that
\begin{align}
    & \P \{ \bar{d}(\bm{X}, g_n ( f_n ( \bm{y})) ) \leq \Delta   \} \nonumber \\
    & \ndot{\leq}  \max_{ \substack{ Q_{X | Y \hat{X}}: \\ \E_Q[d( X, \hat{X} )] \leq \Delta}  } 2^{-nD( Q_{X | Y \hat{X}} \| P_{X|Y} | Q_{Y \hat{X} } ) }.
\end{align}
On the other hand, for every \(\bm{y} \in \mathcal{G}^{\star}( Q_{Y\hat{X}} | Q_Y ) \subset \mathcal{T}_n(Q_Y)\),
\begin{equation}
    \P \{ \bm{Y} = \bm{y}\} = 2^{-n(D(Q_Y \| P_Y) + H( Q_Y ))}.
\end{equation}
Therefore,
\begin{align}
    & \sum_{Q_{ \hat{X} | Y}}  \sum_{\bm{y} \in \mathcal{G}^{\star}( Q_{Y\hat{X}} | Q_Y ) } \! \! \!  \P \{ \bm{Y} = \bm{y}\} \times \nonumber \\
    & \hspace{3.5cm} \P \{ \bar{d}(\bm{X}, g_n ( f_n ( \bm{y})) ) \leq \Delta   \} \nonumber \\
    & \ndot{\leq} \max_{Q_{\hat{X} | Y}} | \mathcal{G}^{\star}( Q_{Y\hat{X}} | Q_Y ) | \times 2^{-n(D(Q_Y \| P_Y) + H( Q_Y ))} \times \nonumber \\
    & \hspace{2cm} \max_{ \substack{ Q_{X | Y \hat{X}}: \\ \E_Q[d( X, \hat{X} )] \leq \Delta}  } 2^{-nD( Q_{X | Y \hat{X}} \| P_{X|Y} | Q_{Y \hat{X} } ) } \\
    & \ndot{\leq} \hspace{-7pt} \max_{ \substack{ Q_{\hat{X} | Y}, Q_{X | Y \hat{X}}: \\ \E_Q[d( X, \hat{X} )] \leq \Delta}  } \hspace{-10pt} \exp \{ - n( D(Q_Y \| P_Y) \!  +  | I_Q(Y; \hat{X}) -R |^{+} + \nonumber \\
    & \hspace{4cm} D( Q_{X | Y \hat{X}} \| P_{X|Y} | Q_{Y\hat{X}} )    ) \}   \label{eq:converse/use_unified_bound} \\
    & = \hspace{-7pt} \max_{ \substack{ Q_{X | Y}, Q_{\hat{X} | Y X}: \\ \E_Q[d( X, \hat{X} )] \leq \Delta}  } \hspace{-10pt} \exp \{ - n( D(Q_{XY} \| P_{XY})   + | I_Q(Y; \hat{X}) -R |^{+}   \nonumber \\
    & \hspace{5.2cm} + I_Q(X;\hat{X} | Y)   ) \} \label{eq:converse/mutual_information_identity},
\end{align}
where  \eqref{eq:converse/use_unified_bound} is due to \eqref{eq:converse/unified_size_upper_bound};
in \eqref{eq:converse/mutual_information_identity} we notice that optimizing over \( ( Q_{\hat{X} | Y}, Q_{X | Y \hat{X}} )\) is the same as optimizing over \((Q_{X | Y}, Q_{\hat{X} | Y X})\) and also utilize \eqref{eq:achievability/mutual_information_identity_presentation}.
The proof is completed after taking the maximization over \(Q_Y\) into account.

\begin{remark}
If we restrict our attention to the case \(R \geq \log | \hat{ \mathcal{X} }|\), i.e., Corollary \ref{cor:no_compression}, the converse proof can be further simplified.
For any coding scheme \((f_n, g_n)\), consider a \(\bm{y} \in \mathcal{T}_n(Q_Y)\) and let \(\hat{\bm{x}} = g_n(f_n(\bm{y}))\).
Since there is no rate constraint and \(\hat{\bm{x}}\) can be any sequence, for any \(\bm{y} \in \mathcal{T}_n(Q_Y)\), we can directly obtain
\begin{align}
    & \P \{ \bar{d}(\bm{X}, g_n ( f_n ( \bm{y})) ) \leq \Delta   \} \nonumber \\
    & \ndot{\leq}  \max_{ \substack{ Q_{X | Y \hat{X}}: \\ \E_Q[d( X, \hat{X} )] \leq \Delta}  } 2^{-nD( Q_{X | Y \hat{X}} \| P_{X|Y} | \hat{P}_{\bm{y} \hat{\bm{x}}} ) } \\
    & \leq  \max_{ \substack{ Q_{\hat{X} | Y},  Q_{X | Y \hat{X}}: \\ \E_Q[d( X, \hat{X} )] \leq \Delta}  } 2^{-nD( Q_{X | Y \hat{X}} \| P_{X|Y} | Q_{Y \hat{X}} ) },
\end{align}
where \(Q_{Y\hat{X}} = Q_Y \times Q_{\hat{X} | Y}\).
Therefore, we do not need to consider \(\mathcal{G}^{\star}( Q_{Y\hat{X}} | Q_Y )\) and the sum over \(\mathcal{T}_n(Q_Y)\) is sufficient.
\end{remark}

\section{Continuity and Rate Regime}
\label{sec:rate_regime}
This section is dedicated to proving Proposition \ref{prop:continuity_and_rate_regime}.
Considering the identity \(|a|^{+} = \max_{\rho \in [0,1]}\rho a\), we can write
\begin{align}
    & E(R, \Delta) \nonumber \\
    & = \min_{Q_{XY}} \min_{ \substack{ Q_{\hat{X} | X Y}: \\ \E_Q[ d(X, \hat{X}) ] \leq \Delta}  } \! \! \! \! D(Q_{XY} \| P_{XY}) +  I_Q(X;\hat{X} | Y) + \nonumber \\
    & \hspace{5cm} |I_Q(Y;\hat{X}) - R|^{+} \\
    & = \min_{Q_{XY}} \min_{ \substack{ Q_{\hat{X} | X Y}: \\ \E_Q[ d(X, \hat{X}) ] \leq \Delta}  }  \! \! \! \! \max_{\rho \in [0,1]} D(Q_{XY} \| P_{XY}) +  I_Q(X;\hat{X} | Y) + \nonumber \\
    & \hspace{5cm} \rho (I_Q(Y;\hat{X}) - R). \label{eq:rate_regime/exponent_formula}
\end{align}
We establish continuity though the following result.
\begin{lemma}
\label{lem:continuity}
Let \(f(s)\) and \(g(s)\) be real-valued functions defined on the same domain \(\mathcal{S}\). Then,
\begin{align}
    | \sup_{s \in \mathcal{S}} f(s) - \sup_{s \in \mathcal{S}}g(s) | & \leq \sup_{s \in \mathcal{S}} |f(s) - g(s) | \label{eq:continuity/sup} \\
    | \inf_{s \in \mathcal{S}} f(s) - \inf_{s \in \mathcal{S}}g(s) | & \leq \sup_{s \in \mathcal{S}} |f(s) - g(s) |. \label{eq:continuity/inf}
\end{align}
\end{lemma}
\begin{proof}
    \eqref{eq:continuity/sup} is easy to verify and \eqref{eq:continuity/inf} is obtained through replacing \(f(s)\) and \(g(s)\) with \(-f(s)\) and \(-g(s)\) in \eqref{eq:continuity/sup}.
\end{proof}
Applying Lemma \ref{lem:continuity} to \eqref{eq:rate_regime/exponent_formula} three times, i.e., applying it in turn to \(Q_{XY}, Q_{\hat{X} | XY}\), and \( \rho \in [0,1]\), we can assert that
\begin{equation}
    |E(R_1, \Delta) - E(R_2, \Delta) | \leq |R_1 - R_2|,
\end{equation}
which implies continuity.
We next investigate the rate regime yielding  \(E(R, \Delta)> 0\).
The following proof is reminiscent of \cite{merhavGeneralizedStochasticLikelihood2017} (see also \cite{scarlettLikelihoodDecoderError2015}).
From \eqref{eq:rate_regime/exponent_formula}, \(E(R, \Delta) > 0\) means that for every \((Q_{XY}, Q_{\hat{X} | XY})\) there exists a \(\rho \in [0,1]\) such that
\begin{equation}
    D(Q_{XY} \| P_{XY}) +  I_Q(X;\hat{X} | Y) + \rho (I_Q(Y;\hat{X}) - R) > 0,
\end{equation}
i.e.,
\begin{equation}
    R < I_Q(Y;\hat{X}) + \frac{1}{\rho} \left(   D(Q_{XY} \| P_{XY}) +  I_Q(X;\hat{X} | Y)   \right). \label{eq:rate_regime/rate_upper_bound}
\end{equation}
Therefore, \(E(R, \Delta) > 0\) is equivalent to
\begin{align}
    & R < \min_{Q_{XY}} \min_{ \substack{ Q_{\hat{X} | X Y}: \\ E_Q[ d(X, \hat{X}) ] \leq \Delta}  }  \! \! \! \! \max_{\rho \in [0,1]} I_Q(Y;\hat{X}) + \nonumber \\
    &  \hspace{2cm} \frac{1}{\rho} \left(   D(Q_{XY} \| P_{XY}) +  I_Q(X;\hat{X} | Y)   \right) \label{eq:rate_regime/divided_by_rho} \\
    & = \min_{ \substack{ P_{\hat{X} | Y}: \\ \E[ d(X, \hat{X}) ] \leq \Delta}  } \! \! \! \! I(Y;\hat{X} ), \label{eq:rate_regime/rate_disortion_function} \\
    & = R_{\mathrm{r}}(P_{XY}, \Delta),
\end{align}
where in \eqref{eq:rate_regime/divided_by_rho} the minimization is because \eqref{eq:rate_regime/rate_upper_bound} holds for every \((Q_{XY}, Q_{\hat{X} | XY})\) and the maximization is due to the existence of \(\rho\); in \eqref{eq:rate_regime/rate_disortion_function} we notice that the minimum in \eqref{eq:rate_regime/divided_by_rho} is achieved if and only if \(D(Q_{XY} \| P_{XY}) +  I_Q(X;\hat{X} | Y) =0\) due to the maximization over \(\rho \in [0,1]\) (in particular, \(\rho = 0\)), which gives rise to \(Q_{XY} = P_{XY}\) and the Markov chain \(X \to Y \to \hat{X}\).

\section{Concluding Remarks}
We investigated the exponential strong converse for remote lossy source coding and characterized the exponent.
We also showed its continuity in \(R\)  through Lemma \ref{lem:continuity} and strict positivity when \(R\) is below the rate-distortion function.
Two previous results are recovered, and our converse proof serves as a simpler alternative proof for \cite[Theorem 1]{kangDeceptionSideInformation2015}.

\section*{Acknowledgment}
This work was partially supported by the European Research Council (ERC) under the ERC Starting Grant N. 101116550 (IT-JCAS).

\atColsEnd{\vskip1pt}

\bibliography{ref}
\bibliographystyle{IEEEtran}

\end{document}